\documentclass[runningheads,a4paper]{llncs}

\usepackage{amssymb,amsmath,wasysym}
\usepackage{graphicx,epsfig}
\usepackage{alltt}
\usepackage{framed}
\usepackage{dsfont}

\usepackage{url}

\newcommand{\keywords}[1]{\par\addvspace\baselineskip
\noindent\keywordname\enspace\ignorespaces#1}

\let\phi\varphi
\let\epsilon\varepsilon

\newcommand{\bvp}[2]{\boxed{\begin{array}{l}#1\\#2\end{array}}}

\newcommand{\idop}{\text{\scshape\texttt{intdiffop}}}
\newcommand{\dop}{\text{\scshape\texttt{diffop}}}
\newcommand{\iop}{\text{\scshape\texttt{intop}}}
\newcommand{\itt}{\text{\scshape\texttt{intterm}}}
\newcommand{\bop}{\text{\scshape\texttt{boundop}}}
\newcommand{\evop}{\text{\scshape\texttt{evop}}}
\newcommand{\edop}{\text{\scshape\texttt{evdiffop}}}
\newcommand{\eiop}{\text{\scshape\texttt{evintop}}}
\newcommand{\eitt}{\text{\scshape\texttt{evintterm}}}
\newcommand{\bp}{\text{\scshape\texttt{bp}}}
\newcommand{\bc}{\text{\scshape\texttt{bc}}}
\newcommand{\gbp}{\text{\scshape\texttt{gbp}}}
\newcommand{\es}{\text{\scshape\texttt{es}}}

\newcommand{\D}[1]{\operatorname{D}^{#1}}
\newcommand{\Ev}[1]{\operatorname{E}\lbrack#1\rbrack}

\newcommand{\A}{\operatorname{A}}
\newcommand{\BC}{\operatorname{BC}}
\newcommand{\BP}{\operatorname{BP}}

\newcommand{\N}{\mathbb{N}}
\newcommand{\R}{\mathbb{R}}

\newcommand{\E}{\text{\textsc{E}}}

\newcommand{\der}{\partial}

\newcommand{\cum}{{\textstyle \varint}}

\newcommand{\galg}{\mathcal{F}}

\newcommand{\V}{\mathcal{F}}
\newcommand{\W}{\mathcal{G}}

\newcommand{\f}{f}        

\newcommand{\B}{\mathcal{B}}
\newcommand{\be}{\beta}
\newcommand{\Bo}{\B^{\perp}}
\newcommand{\C}{\mathcal{E}}

\newcommand{\fri}[1]{#1^\blacklozenge}

\DeclareMathOperator{\Ker}{Ker}
\DeclareMathOperator{\Img}{Im}
\DeclareMathOperator{\codim}{codim}

\newcommand{\tma}{TH$\exists$OREM$\forall$}

\newcommand{\evl}{\text{\scshape\texttt e}}

\newcommand{\ocum}{{\setbox0=%
    \hbox{$\textstyle{\scriptstyle-}{\varint}$}%
    \textstyle{\vcenter{\hbox{$\scriptstyle-$}}\kern-.5\wd0}%
    \!\varint}}

\newcommand{\intdiffop}{\galg[\der,\cum]}

\newcommand{\dirs}{\dotplus}

\newcommand{\id}{1}

\newcommand{\maple}{\textsc{Maple}}

\begin{document}

\mainmatter
\title{Regular and Singular Boundary Problems \\in \maple}

\author{Anja Korporal\inst{1}\thanks{Partially supported by the RISC PhD scholarship program of the government of Upper Austria} \and Georg Regensburger\inst{2}\thanks{Supported by the Austrian Science Fund (FWF): J 3030-N18.} \and Markus Rosenkranz\inst{3} }
\authorrunning{A. Korporal,  G. Regensburger, M. Rosenkranz}

\institute{ Johann Radon Institute for Computational and Applied Mathematics,\\
Austrian Academy of Sciences, Altenberger Str. 69, 4040 Linz, Austria\\[1ex]
\and
INRIA Saclay -- \^{I}le de France, Project DISCO, L2S, \\ Sup\'{e}lec, 91192 Gif-sur-Yvette Cedex, France\\[1ex]
\and
University of Kent, \\
Cornwallis Building, Canterbury, Kent CT2 7NF, United Kingdom
}

\maketitle

\begin{abstract}
  We describe a new \maple\ package for treating boundary problems for
  linear ordinary differential equations, allowing two-/multi-point as
  well as Stieltjes boundary conditions. For expressing differential
  operators, boundary conditions, and Green's operators, we employ the
  algebra of integro-differential operators. The operations
  implemented for regular boundary problems include computing Green's
  operators as well as composing and factoring boundary problems. Our
  symbolic approach to singular boundary problems is new; it provides
  algorithms for computing
compatibility conditions and generalized Green's operators.
\end{abstract}


\section{Introduction}
\label{sec:intro}
Although boundary problems clearly play an important role in
applications and in Scientific Computing, there is no systematic
support for solving them symbolically in current computer algebra
systems. In this paper, we describe a \maple\ package with algorithms
for regular as well as singular boundary problems for linear ordinary
differential equations (LODEs). While a first version of the package
with functions for regular boundary problems was presented
in~\cite{KorporalRegensburgerRosenkranz2010}, the methods and the
implementation for singular problems are new. A prototype
implementation for regular boundary problems in the \tma\ system was
described in \cite{RosenkranzRegensburgerTecBuchberger2009} as part of
a general symbolic framework for boundary problems, including also
some first steps towards linear partial differential equations
(LPDEs).

 In Section \ref{sec:idop},  we recall the algebra of
integro-differential operators providing the algebraic structure for
computing with boundary problems. We describe its implementation in
\maple, where we use a normal form approach in contrast to
\cite{RosenkranzRegensburgerTecBuchberger2009}.  In Section
\ref{sec:reg}, we outline our symbolic approach for solving boundary
problems. For an analytic treatment of boundary problems for LODEs,
see for example~\cite{CoddingtonLevinson1955,Stakgold1979} or
\cite{Duffy2001} for further applications. The functions we present
include the computation of Green's operators and Green's functions as
well as the factorization of boundary problems.

We introduce generalized boundary problems in Section \ref{sec:sing}
and develop an algorithm for computing generalized Green's
operators. The main step of the algorithm is to determine
compatibility conditions for arbitrary boundary problems in an
algebraic setting; the special case of two-point boundary problems of
second order is discussed in \cite[Lecture 34]{Agarwal2008}. For
singular boundary problems and generalized or modified Green's
functions in Analysis, we refer for example to~\cite{Stakgold1979}
and~\cite{Loud1970}, and in the context of generalized inverses to
\cite[Sect.~9.4]{Ben-Israel2003}, \cite{Boichuk2004},
and~\cite[Sect.~H]{Nashed1976}.

The \maple\ package \emph{IntDiffOp} is
available with an example worksheet at
\url{http://www.risc.jku.at/people/akorpora/index.html}.

\section{Integro-Differential Operators}
\label{sec:idop}

We first recall the definition of integro-differential algebras and
operators, see \cite{RosenkranzRegensburger2008a} and
\cite{RosenkranzRegensburgerTecBuchberger2011} for further
details. For the similar notion of differential Rota-Baxter algebras,
we refer to~\cite{GuoKeigher2008}. As a motivating example, consider
the algebra $\galg = C^\infty (\mathbb{R})$ with the usual derivation
and the integral operator $\cum\colon f \mapsto \int_a^x f(\xi) \,
d\xi$ for a fixed $a \in \mathbb{R}$. The essential algebraic
identities satisfied by the derivation and the integral operator are
the Leibniz rule, the Fundamental Theorem of Calculus, and Integration
by Parts. Note also that $f(a)=f-\cum f'$, so the evaluation
$\evl_a\colon f \mapsto f(a)$ at the initialization point $a$ of the
integral can also be expressed in terms of the derivation and
integral.

We call $(\galg, \partial, \cum)$ an \emph{integro-differential
  algebra} if $(\galg,\der)$ is a commutative differential algebra
over a commutative ring $K$ and $\cum$ is a $K$-linear right inverse
(section) of $\der={}'$, meaning~$(\cum f)' = f$, such that the
\emph{differential Baxter axiom}
\begin{equation*}
\label{eq:diff-baxter-axiom}
(\cum f')(\cum g') + \cum (f g)' = (\cum f') g + f (\cum g')
\end{equation*}
holds. We call $\evl = 1-\cum \circ \der$ the \emph{evaluation} of
$\galg$. We say that an integro-differential algebra over a field $K$
is ordinary if $\Ker (\der)=K$. For an ordinary integro-differential
algebra, the evaluation can be interpreted as a multiplicative linear
functional (character) $\evl \colon \galg \rightarrow K$. This allows
treating initial value problems, but for doing boundary problems we
need additional characters $\varphi \colon \galg \rightarrow K$ (in
the above example, evaluations $\evl_c\colon f \mapsto f(c)$ at various
points $c \in \mathbb{R}$).

Let $(\galg, \der, \cum)$ be an ordinary integro-differential algebra over a field $K$ and let $\Phi \subseteq \galg^*$  be a set of multiplicative linear functionals $\phi \colon \galg \rightarrow K$ including $\evl$. The \emph{integro-differential operators} $\galg_{\Phi}[\der, \cum]$ are defined in \cite{RosenkranzRegensburger2008a}
 as the $K$-algebra generated by the symbols $\der$ and $\cum$, the
 ``functions'' $f\in \galg$ and the ``functionals''
 $\phi \in \Phi$, modulo the Noetherian and confluent rewrite system of Table~\ref{fig:red-rules}.

\begin{table}[h]
  \newcommand{\ra}{\rightarrow}
  \centering
  \renewcommand{\baselinestretch}{1.3}
  \small
  \begin{tabular}[h]{|@{\quad}lcl@{\quad}|@{\quad}lcl@{\quad}|@{\quad}lcl@{\quad}|}
    \hline
    $f g$ & $\ra$ & $f \cdot g$ &
      $\der f$ & $\ra$ & $f \der + f'$ &
      $\cum f \cum$ & $\ra$ & $(\cum f) \, \cum - \cum \, (\cum
        f)$\\
    $\phi \psi$ & $\ra$ & $\psi$ &
      $\der \phi$ & $\ra$ & $0$ &
      $\cum f \der$ & $\ra$ & $f - \cum f' - \evl(f)
         \, \evl$\\
    $\phi f$ & $\ra$ & $\phi (f) \, \phi$ &
      $\der\cum$ & $\ra$ & $1$ &
      $\cum f \phi$ & $\ra$ & $(\cum f) \, \phi$\\
    \hline
  \end{tabular}
  \medskip
  \caption{Rewrite Rules for Integro-Differential Operators}
  \label{fig:red-rules}
\end{table}

The representation of integro-differential operators in our \maple\ implementation is based on the fact that every integro-differential operator has a unique normal form as a sum of a differential, integral, and boundary operator.
 The normal forms of differential operators are as usual $\sum f_i\der^i$, integral operators can be written uniquely (up to bilinearity) as sums of terms of the form $f\cum g$, and the normal forms of \emph{boundary operators} are given by
\begin{equation}
\label{eq:BOP}
\sum_{\varphi \in \Phi} \Bigl( \sum_{i \in \N} f_{i, \varphi} \varphi \partial^i +
 \sum_{j\in \N} g_{j, \varphi} \phi \cum h_{j, \varphi} \Bigr),
\end{equation}
with only finitely nonzero summands. Stieltjes \emph{boundary
  conditions} are boundary operators where $f_{i, \varphi}=a_{\phi,i}
\in K$ and $g_{j,\phi} =1$. They act on $\galg$ as linear functionals
in the dual space $\V^*$. See \cite{BrownKrall1974} for Stieltjes
boundary conditions in Analysis.

From Table~\ref{fig:red-rules} formulas can be derived for expressing
the product of integro-differential operators directly in terms of
normal forms; see \cite{RegensburgerRosenkranzMiddeke2009} for the
case $\Phi=\{\evl\}$. Implementing these formulas leads to faster
computations since we need not reduce in each step. In our package, we
use for the underlying ``integro-differential algebra'' all the smooth
functions in one variable representable in \maple, together with the
usual derivation and the integral operator $\cum = \cum_0^x$, both
computed by \maple\ . We take as characters $\Phi = \{ \evl_c \mid c
\in \R\}$.

We created data types for the different kinds of operator,
representing integro-differential operators as triples $\idop(a,b,c)$,
where $a$ is a differential operator, $b$ an integral operator and $c$
a boundary operator. Differential operators are represented as lists
$\dop(f_0, f_1, \ldots)$ and integral operators as lists of pairs of
the form $\iop(\itt(f_1, g_1), \itt(f_2, g_2), \ldots)$. In order to
have a unique representation for integral operators, one would need a
basis of the underlying integro-differential algebra and use only
basis elements for the $g_i$. In our implementation, we use the following
heuristic approach: We split sums in the $g_i$ and move scalar
factors to the coefficients $f_i$.

Due to \eqref{eq:BOP}, a boundary operator $\bop$ contains a
list of evaluations at different points. Each evaluation $\evop$ is a
triple containing the evaluation point, the local part $\sum
f_{i, \varphi} \varphi \partial^i$ and the global part $\sum g_{j, \varphi} \varphi \cum h_{j, \varphi}$. Hence we use
the expression $\bop(\evop(c, \edop(f_0, \ldots), \eiop(\eitt(g_1,
h_1), \ldots), \ldots)$ for the representation of boundary
operators.

In the following example, we first enter some operators of
different types. For displaying the operators, we use $\D{}$ for $\partial$,
$\A$ for $\cum$ and $\Ev{c}$ for the evaluation $\evl_c$.

\newcommand{\negskip}{~\vspace*{-3ex}}
\newcommand{\negsskip}{~\vspace*{-2ex}}
\begin{center}
\begin{minipage}{0.999\textwidth}
\begin{framed}~\\[-1cm]\small
\begin{alltt}
> T := DIFFOP(0,0,1);\negsskip
        \[ T := \D{2} \] \negskip
> G := INTOP(INTTERM(1,1));\negsskip
        \[ G := \A \] \negskip
> B := BOUNDOP(EVOP(1, EVDIFFOP(1), EVINTOP(EVINTTERM(1,1))));\negsskip
\[  B := \Ev{1}  + ((\Ev{1}) . A) \] \negskip \negskip
\end{alltt}\normalsize
\end{framed}
\end{minipage}
\end{center}

Now we show how to add and multiply integro-differential operators and how to apply them to a function $f \in \V$.

\begin{center}
\begin{minipage}{0.999\textwidth}
\begin{framed}~\\[-1cm]\small
\begin{alltt}
> ApplyOperator(G, f(x));
\end{alltt}
\begin{equation*}
         \cum_0^x f(x) dx
\end{equation*}
\begin{alltt}
> MultiplyOperator(G,G);\negsskip
   \[(x . A) - (A . x)\] \negskip
> MultiplyOperator(T,G,G);\negsskip
 \[ 1\] \negskip
> S := AddOperator(T, G, B);
 \[ S := \D{2} + A + \Ev{1}+ ((\Ev{1}) . A) \] \negskip
> ApplyOperator(S, f(x));
\end{alltt}
\begin{equation*}
  \frac{\operatorname{d}^2}{\operatorname{dx}^2} f(x) +  \cum_0^x f(x) dx + f(1) +  \cum_0^1 f(x) dx
\end{equation*}
~\\[-0.9cm]
\end{framed}\normalsize
\end{minipage}
\end{center}

\section{Regular Boundary Problems in \maple}
\label{sec:reg}
In this section, we demonstrate how to compute with regular boundary problems
in our \maple\ package.
For an integro-differential algebra $\V$, a boundary problem is given by a monic differential operator $T = \partial ^n + c_{n-1} \partial^{n-1} + \dots + c_1 \partial + c_0$ and boundary conditions $\beta_1,\ldots , \beta_m$.
Given a forcing function $f \in \V$, we want to find $u \in \V$ such that
\begin{equation}\label{BP:General}
   \bvp{Tu=f,}{\beta_1 u = \cdots = \beta_n u =0.}
\end{equation}
A boundary problem is called regular if for each $f \in \V$ there is
exactly one $u \in \V$ satisfying \eqref{BP:General}. We want to solve
a boundary problem not only for a fixed~$f$ but to compute the Green's
operator mapping each forcing function~$f$ to its unique
solution~$u$. In other words, we solve a whole family of
inhomogeneous differential equations, parameterized by a ``symbolic''
right-hand side~$f$. We restrict ourselves
to homogeneous conditions because the general solution is then obtained by
adding a particular solution satisfying the inhomogeneous conditions.

For convenience, we shortly recall the abstract linear algebra setting
for boundary problems over a vector space $\V$ as described in
\cite{RegensburgerRosenkranz2009}. For $U \leq \V$ we define the \emph{orthogonal} as $U^{\perp} = \{ \beta \in \V^* : \beta(u) = 0 \text{ for all } u \in U\} \leq \V^*$. Similarly, for $\B \leq \V^*$, we define $\Bo = \{v \in \V : \beta(v) = 0 \text{ for all } \beta \in \B\}\leq \V$. A subspace $U$ (resp. $\B$) is \emph{orthogonally closed} if $U = U^{\perp\perp}$ (resp. $\B = \B^{\perp\perp}$). Every subspace $U \leq \V$ is orthogonally closed and every finite dimensional subspace $\B \leq \V^*$ is orthogonally closed. For a linear map $T \colon \V \to \W$ between vector spaces, the transpose map $T^* \colon  \W^* \to \V^*$ is defined by $\gamma \mapsto \gamma \circ T$. The image of an orthogonally closed space under the transpose map is orthogonally closed.

A \emph{boundary problem} is given by a pair $(T,\B)$, where $T$ is a surjective
linear map and $\B \le \V^*$ is an orthogonally closed subspace of the
dual space.
We call $u \in \V$ a solution of
$(T,\B)$ for a given $\f \in \V$ if $Tu=\f$ and $u\in \Bo$. A boundary
problem is \emph{regular} if for each $f$ there exists a unique
solution $u$. The \emph{Green's operator} of a regular problem maps
each $f$ to its unique solution $u$. We also write
$(T, \B)^{-1}$ for the Green's operator. A boundary problem is
\emph{regular} iff $\B^{\perp}$ is a complement of $\Ker T$ so that
$\V=\Ker T \dirs \B^{\perp}$ as a direct sum.

For $\V=C^\infty[a,b]$, a monic differential operator $T$ is always
surjective and $\dim \Ker T =n < \infty$.  Moreover, variation of
constants can be used to compute a distinguished right inverse: If~$T$
has order~$n$ and $u_1, \ldots, u_n$ is a fundamental system for it,
the \emph{fundamental right inverse} is given by
\begin{equation}\label{Eq:FRI}
   \fri{T} = \sum_{i=1}^n u_i \cum d^{-1} d_i,
\end{equation}
where $d$ is the determinant of the Wronskian matrix $W$ for $(u_1,
\ldots, u_n)$ and $d_i$ the determinant of the matrix $W_i$ obtained
from $W$ by replacing the $i$-th column by the $n$-th unit vector.
Equation \eqref{Eq:FRI} is valid in arbitrary integro-differential
algebras provided the $n$-th order operator $T$ has a fundamental
system $(u_1, \ldots, u_n)$ with invertible Wronskian matrix; see
\cite{RosenkranzRegensburger2008a} or
\cite{RosenkranzRegensburgerTecBuchberger2011}. This will be
assumed from now on, together with the condition $\dim\B < \infty$
appropriate for LODEs.

Regularity of a boundary problem $(T, \B)$ can be tested
algorithmically as follows.
If $(u_1, \ldots, u_n)$ is a basis
for~$\Ker{T}$ and $(\beta_1, \ldots, \beta_m)$ for~$\B$, we have a
regular problem iff the \emph{evaluation matrix}
\begin{equation}
  \label{eq:EvalMat}
  \beta(u)=
  \begin{pmatrix}
    \beta_1(u_1) & \dots & \beta_1(u_n) \\
    \vdots & \ddots & \vdots \\
    \beta_m(u_1) & \dots &\beta_m(u_n)
  \end{pmatrix}
\end{equation}
is regular; see \cite[Cor. A.17]{RegensburgerRosenkranz2009} or
\cite[p.~184]{Kamke1967} for the special case of two-point boundary
conditions. Of course this implies $m=n$, but we will consider more
general types of boundary problems in Section~\ref{sec:sing} where
this is no longer the case. It will also be convenient to use the
notation~\eqref{eq:EvalMat} for arbitrary~$u_1, \ldots, u_n \in \V$
and boundary conditions~$\beta_1, \ldots, \beta_m$.

The algorithm for computing the Green's operator is described in
detail in \cite{RosenkranzRegensburger2008a}; see also
\cite{RosenkranzRegensburgerTecBuchberger2009}.  The main steps
consist in computing the fundamental right inverse $\fri{T} \in
\intdiffop$ from a given fundamental system as in \eqref{Eq:FRI} and
the projector $P \in \intdiffop$ onto $\Ker T$ along $\Bo$. Then the
Green's operator is then computed as $G=(\id -P)\fri{T}$.

For a boundary problem we need to enter a monic differential operator $T$ and
a list of boundary conditions $(b1, \ldots, bm)$ as
described in Section \ref{sec:idop} in the form $\bp(T, \bc(b1,
\ldots, bm))$. We use the
\maple\ function \emph{dsolve} for computing a fundamental system of $T$.  As an example, we compute
the Green's operator for the simplest two-point boundary problem $u''
= f$, $u(0)=u(1)=0$. From the Green's operator for two-point boundary problems, we can extract the Green's function~\cite{Rosenkranz2005}, which is usually used in Analysis to represent the Green's operator.

\begin{center}
\begin{minipage}{0.95\textwidth}
\begin{framed}~\\[-1cm]\small
\begin{alltt}
> T := DIFFOP(0,0,1):
> b1 := BOUNDOP(EVOP(0, EVDIFFOP(1), EVINTOP())):
> b2 := BOUNDOP(EVOP(1, EVDIFFOP(1), EVINTOP())):
> Bp := BP(T, BC(b1, b2));
\[ \mathit{Bp} := \operatorname{BP}(\D2, \operatorname{BC}(\Ev0, \Ev1)) \] \negskip
> IsRegular(Bp);      \negskip
\[ true \] \negskip
> GreensOperator(Bp); \negskip
\[ (x . \A) - (\A . x) - ((x \Ev1) . \A) + ((x \Ev1) . \A . x)  \] \negskip
> GreensFunction(%);
\end{alltt}
\begin{equation*}
   \begin{cases}
      -\xi + x \xi    &    0 <= \xi \text{ and } \xi <= x \text{ and } x <= 1 \\
      -x + x \xi     &    0 <= x \text{ and } x <= \xi \text{ and } \xi <= 1
   \end{cases}
\end{equation*}
~\\[-0.8cm]
\end{framed}\normalsize
\end{minipage}
\end{center}

For simplifying boundary problems, we can apply factorizations into
lower order problems along given factorizations of the differential
operators.  Further details and proofs of the following results can be
found in \cite{RegensburgerRosenkranz2009} and
\cite{RosenkranzRegensburger2008a}. The composition of two boundary
problems $(T_1, \B_1)$ and $(T_2, \B_2)$ is defined as
\begin{equation} \label{Eq:BComp}
(T_1, \B_1) \circ (T_2, \B_2) = (T_1T_2, T_2^*(\B_1) + \B_2).
\end{equation}
The composition $(T_1, \B_1) \circ (T_2, \B_2)$ of two regular boundary problems is regular with Green's operator
\begin{equation} \label{Eq:GComp}
((T_1, \B_1) \circ (T_2, \B_2))^{-1} = (T_2, \B_2)^{-1} (T_1, \B_1)^{-1}.
\end{equation}
Given a regular boundary problem $(T, \B)$, every factorization
$T=T_1T_2$ can be lifted to a factorization $(T, \B) = (T_1, \B_1)
\circ (T_2, \B_2)$, where $(T_1, \B_1)$ and $(T_2, \B_2)$ are regular
and $\B_2 \leq \B$. For factorizing a differential operator, we use the
function \emph{DFactor} in the \maple\ package \emph{DEtools}. As an
easy example, we show how to factor the boundary problem from above;
more examples for solving and factoring boundary problems can be found in our example worksheet.

\begin{center}
\begin{minipage}{0.95\textwidth}
\begin{framed}~\\[-1cm]\small
\begin{alltt}
> Bp := BP(T, BC(b1, b2));
  \[ \mathit{Bp} := \BP(\D{2}, \BC(\Ev{0}, \Ev{1})) \] \negsskip
> f1, f2 := FactorBoundaryProblem(Bp);
   \[ f1,\; f2 := \BP(\D{}, \BC(\Ev1 . \A)),\; \BP(\D{}, \BC(\Ev0))  \]
\end{alltt}\normalsize
~\\[-1.7cm]
\end{framed}
\end{minipage}
\end{center}

\section{Singular Boundary Problems}
\label{sec:sing}

For illustrating the main issues with singular boundary problems, we
consider the boundary problem
\begin{equation} \label{BP:ExSing}
   \bvp{u'' = f,}{u'(0)=u'(1)=0;}
\end{equation}
see for example \cite[Page 215]{Stakgold1979} or \cite[Section
3.5]{Rosenkranz2005} from a Symbolic Computation perspective.  This
problem is singular since it is not solvable for all $f \in \V$. It
can easily be seen that if $u'' = f$, then $f$ has to fulfill the
\emph{compatibility condition} $u'(1) = \int_0^1 f(\xi) \,
d\xi=0$. Moreover, uniqueness fails as well: If a solution $u \in \V$
exists, then also $u +c$ solves the problem for all $c \in \R$.

Our goal here is to generalize the symbolic
approach of the previous section to problems of the
kind~\eqref{BP:ExSing}. Since we want to compute generalized Green's
operators, we cannot give up uniqueness of solutions---but we no
longer require existence. Of course, uniqueness of solutions can
always be achieved by imposing additional boundary conditions. On the
other hand, adding too many conditions introduces new compatibility
conditions, which we want to avoid (see after
Lemma~\ref{EvMatFullColumnRank} for the precise statement).  For the
boundary problem~\eqref{BP:ExSing}, we can add for example the
condition $u(1)=0$ and consider the problem
\begin{equation} \label{BP:ExSingPlus}
   \bvp{u'' = f,}{u'(0)=u'(1)=u(1)=0.}
\end{equation}
This does not introduce any new compatibility conditions as we will
see later (see before Lemma~\ref{Embedding}).

A boundary problem has at most one solution for each forcing function
$f$ iff $\Bo \cap \Ker T = \{0\}$. We see that for \eqref{BP:ExSing}
we have $\Bo \cap \Ker T = \R$ while in \eqref{BP:ExSingPlus} the
intersection is $\{0\}$. The regularity test for boundary problems in
terms of the evaluation matrix~\eqref{eq:EvalMat} can be generalized
from the setting in Section~\ref{sec:reg}.

\begin{lemma} \label{EvMatFullColumnRank} Let $U = [u_1, \ldots , u_n]
  \leq \V$ and $\B =[\beta_1, \ldots, \beta_m] \leq \V^*$ with
  $\beta_i$ and $u_j$ linearly independent. Then $U \cap \B ^{\perp} =
  \{0\}$ iff the evaluation matrix $\beta(u)$
   has full column rank.
\end{lemma}

\begin{proof}
  Let $b_j$ denote the columns of $\beta(u)$. The evaluation matrix
  has deficient column rank iff there exists a linear combination
  $\sum_{j=1}^n \lambda_j b_j =0$ with at least one $\lambda_j \not=
  0$. This is the case iff there exist a nonzero $u = \sum_{j=1}^n
  \lambda_j u_j \in U \cap \B_1^{\perp}$. \qed
\end{proof}

As mentioned for the example~\eqref{BP:ExSingPlus}, singular boundary
problems typically impose \emph{compatibility conditions} on the
admissible forcing functions. We can now make this precise: Clearly, a
function~$f$ is admissible iff it is of the form $Tu$ for a function
$u$ that satisfies the boundary conditions from $\B$, so the space of
admissible functions is $T(\Bo)$. The compatibility conditions provide
an implicit description of this space, comprising all those linear
functionals that annihilate $T(\Bo)$. In other words, the
compatibility conditions are the subspace $T(\Bo)^{\perp}$ of
$\V^*$. This also makes precise what we mean by adding boundary
conditions without imposing additional compatibility conditions: We
enlarge~$\B$ to~$\tilde\B$ so as to ensure $\tilde\B^{\perp} \cap
\Ker{T} = \{ 0 \}$ despite retaining $T(\Bo) = T(\tilde\B^{\perp})$.

For tackling the problem of existence, we modify the forcing
function. In the example~\eqref{BP:ExSingPlus}, this looks as follows:
Since a solution exists only for forcing functions that fulfill
$\int_0^1 f(\xi) \, d\xi=0$, we consider the problem
\begin{equation} \label{BP:ExSingPlusQ}
   \bvp{u'' = f- \int_0^1 f(\xi) \, d\xi, }{u'(0)=u'(1)=u(1)=0,}
\end{equation}
which now always has a unique solution. For those $f$ that fulfill the
compatibility condition, problem \eqref{BP:ExSingPlus} remains
unchanged.

The general idea is that we project an arbitrary forcing function into
the space of admissible functions. But this involves choosing those
``exceptional functions'' that we want to filter out. Even in the
simple example~\eqref{BP:ExSingPlus}, we might as well project $f$ to
$f - \frac{1}{2}x \cum_0^1 f(\xi) \, d\xi$ instead of $f - \cum_0^1 f(\xi) \,
d\xi$. In the second case, we have filtered out the constant
functions, in the first case the linear-homogeneous ones. The
space~$\C$ of exceptional functions can be any complement of the
space~$T(\Bo)$ of admissible functions, like $\C = [1]$ or $\C = [x]$
in this example.

\begin{definition}
  \label{Def:GenBndProb}
 A \emph{generalized boundary problem} is given by a triple $(T,\B,
 \C)$, where $(T, \B)$ is a boundary problem and $\C \leq \V$. A
 generalized boundary problem is called \emph{regular} if
\begin{equation*}
\Bo \cap \Ker T = \{0\} \quad \text{and} \quad \V = T(\Bo) \dotplus \C.
\end{equation*}
The \emph{generalized Green's operator} maps each forcing function $f$  to the unique solution of the boundary problem
\begin{equation*}
   \bvp{Tu = Qf,}{\be_1u=\ldots = \be_mu=0,}
\end{equation*}
where $\B = [\be_1, \ldots, \be_m]$ and $Q$ is the projector onto
$T(\Bo)$ along $\C$. We also write $(T, \B, \C)^{-1}$ for the Green's operator.
\end{definition}

If $(T, \B, \C)$ is regular,  the restriction $T|_{\Bo} \colon \Bo \to T(\Bo)$ is bijective.  So the generalized Green's
operator is given by
\begin{equation}
\label{eq:generalizedG}
G = T|_{\Bo}^{-1} Q.
\end{equation}
We begin with computing the projector $Q$. For this we derive first an
explicit description of the space of compatibility conditions.

\begin{proposition}\label{Prop:CompCond}
   Let $(T, \B, \C)$ be a generalized boundary problem and let $G$ be any right inverse of $T$. Then we have
   \begin{equation}
     \label{eq:CompCond}
      T(\Bo)^{\perp} = G^*(\B \cap (\Ker T)^{\perp}).
   \end{equation}
   Moreover, $\dim T(\Bo)^{\perp} = \dim \C$ for any complement $\C$ with $\V = T(\Bo) \dotplus \C$.
\end{proposition}

\begin{proof}
  With \cite[Prop. A.6]{RegensburgerRosenkranz2009}, we see that
  $T(\Bo)^{\perp} = (T^*)^{-1}(\B)$. Since $T$ is surjective, $T^*$ is
  injective, and for any right inverse $G$ of $T$, $G^*$ is a left
  inverse of $T^*$. Hence $(T^*)^{-1}(\B) = G^*(\B \cap \Img T^*)$ by
  \cite[Prop. A.13]{RegensburgerRosenkranz2009}. Again by
  \cite[Prop. A.6]{RegensburgerRosenkranz2009}, we have $\Img
  T^* = (\Ker T)^{\perp}$, and hence
    \begin{math}
      T(\Bo)^{\perp} = G^*(\B \cap (\Ker T)^{\perp}). 
   \end{math}

   Since $\dim \B < \infty$, by the first statement also $\dim
   T(\Bo)^{\perp} < \infty$. But $T(\Bo)$ is orthogonally closed; see
   for example \cite[Section
   A.1]{RegensburgerRosenkranz2009}. Therefore we obtain
   \begin{equation*}
      \dim T(\Bo)^{\perp} = \codim T(\Bo)^{\perp\perp} = \codim T(\Bo),
   \end{equation*}
   and the statement follows immediately from \cite[Prop. A.14]{RegensburgerRosenkranz2009}. \qed
\end{proof}

Note that $s = \dim{\C} = \codim{T(\Bo)}$ counts the number of
(linearly independent) compatibility
conditions. Equation~\eqref{eq:CompCond} is the key for an algorithmic
description of the projector~$Q$ onto~$T(\Bo)$ along~$\C$. The
space~$\C$ is given as part of the problem description, and it can be
specified by a basis $(w_1, \ldots, w_s)$. Since the other space
$T(\Bo)$ has finite codimension $s$, it can be specified in terms
of~$s$ linearly independent compatibility conditions, and
Equation~\eqref{eq:CompCond} can be used to compute these in terms
of~$T$ and~$\B$. For that we just have to determine a basis of~$\B
\cap (\Ker T)^{\perp}$ and then apply any right inverse~$G$ of~$T$,
for example the fundamental right inverse~$\fri{T}$ defined in
Section~\ref{sec:reg}.

For determining a basis of $\B \cap (\Ker T)^{\perp}$ we first compute
the kernel of the transpose of the evaluation matrix $\beta(u)$, where
$(u_1, \ldots , u_n)$ is any basis of $\Ker T$ and $(\beta_1, \ldots,
\beta_m)$ any basis of $\B$. If $w=(w_1, \ldots, w_m)^t \in \Ker
\beta(u)^t$, then
\begin{equation*}
   w^t(\beta_1, \ldots, \beta_m)^t = \sum_{i=1}^m w_i \beta_i \in \B \cap (\Ker T)^{\perp},
\end{equation*}
hence a basis of $\B \cap (\Ker T)^{\perp}$ can be obtained by
computing the products $(v_1^t (\beta_1, \ldots,  \beta_m)^t, \ldots ,
v_k^t (\beta_1, \ldots, \beta_m)^t)$, where $(v_1, \ldots, v_k)$ is a
basis of $\Ker \beta(u)^t$.

Using Proposition \ref{Prop:CompCond}, we can now verify that the
compatibility conditions of the boundary problems \eqref{BP:ExSing}
and \eqref{BP:ExSingPlus} are the same. In both cases we have
$T=\partial^2$, so we can choose the fundamental right inverse $\cum
\cum = x\cum - \cum x$ and $(1,x)$ as a basis of $\Ker T$.
The evaluation matrices are given by
\begin{equation*}
   \beta(u)
   = \begin{pmatrix}
        0 & 1 \\0 & 1
     \end{pmatrix}
   \quad
   \text{and}
   \quad
   \beta(u)
    =
		\begin{pmatrix}
                  0 & 1 \\0 & 1 \\1 & 1
              	\end{pmatrix}.
\end{equation*}
In the first case, a basis of $\beta(u)^t$ is given by $((-1, 1)^t)$, hence
$(\E_1\partial - \E_0\partial)$ is a basis of $\B \cap (\Ker
T)^{\perp}$. In the second case, a basis of $\beta(u)^t$ is given by $((-1,1,0)^t)$ and the basis
of $\B \cap (\Ker T)^{\perp}$ is again $(\E_1\partial -
\E_0\partial)$. Multiplying this basis by the right inverse of~$T$, we
get as a basis for the compatibility conditions
\begin{align*}
  & (\E_1\partial - \E_0\partial) \cdot (x\cum - \cum x)
= \E_1 (x\partial +1)\cum - \E_0 (x\partial +1)\cum - \E_1\partial
\cum x + \E_0\partial \cum x\\
& \qquad = \E_1 x + \E_1 \cum - \E_0 x - \E_0 \cum -\E_1x + \E_0 x =
\E_1 \cum = \cum_0^1,
\end{align*}
which agrees with our heuristic considerations after~\eqref{BP:ExSing}.

We can now compute the
projector $Q$ just as the kernel projector~$P$ for standard boundary
problems (mentioned in Section~\ref{sec:reg}).  If $(\kappa_1, \ldots,
\kappa_s)$ is a basis for the compatibility
conditions~$T(\Bo)^{\perp}$ and $(w_1, \ldots, w_s)$ a basis for~$\C$,
then the corresponding evaluation matrix $\kappa(w)$ is regular by
Lemma \ref{EvMatFullColumnRank}, which can be applied to
$\V = T(\Bo) \dotplus \C = T(\Bo)^{\perp \perp} \dotplus \C$
since $T(\Bo)$ is orthogonally closed. Hence we can compute the
projector $Q$ onto $T(\Bo)$ along $\C$ as
\begin{displaymath}
 Q = 1- \sum_{i=1}^s w_i \tilde{\kappa}_i,
\end{displaymath}
where $(\tilde{\kappa}_1, \ldots, \tilde{\kappa}_s)^t = \kappa(w)^{-1}
\cdot (\kappa_1, \ldots, \kappa_s)^t$; see for example \cite[Lemma
A.1]{RegensburgerRosenkranz2009}.

The final step for computing the generalized Green's operator~\eqref{eq:generalizedG} is to
find the inverse function $T|_{\Bo}^{-1}$. In the regular case, we
started with an arbitrary right inverse of $T$ and multiplied with a
projection onto $\Bo$ along $\Ker T$. But this step cannot be
generalized to our setting. Our approach is to embed the generalized
problem into a standard one in the following sense.

First note that the evaluation matrix of a regular generalized boundary problem has full column rank by  Lemma \ref{EvMatFullColumnRank}, so it has a left inverse.

\begin{lemma}
  \label{Embedding}
  Let $(T, \B, \C)$ be a regular generalized boundary problem.
  Let $\beta(u)^-$ be a left inverse of $\beta(u)$ and
  $(\tilde{\beta}_1, \ldots ,
  \tilde{\beta}_n)^t = \beta(u)^-(\beta_1, \ldots, \beta_m)^t$.
  Then the boundary problem $(T, \tilde{\B})$ is regular, where
  $\tilde{\B} \leq \B$ is spanned by $\tilde{\beta}_1, \ldots ,
  \tilde{\beta}_n$.
\end{lemma}

The proof of the statement is obvious, since
  the evaluation matrix $\tilde{\beta}(u)$ is given
  by $\beta(u)^- \beta(u) = 1_n$. Hence the problem $(T, \tilde{\B})$ is
  regular.
In our package, we always choose the Moore-Penrose pseudoinverse as a
left inverse~$\beta(u)^-$ of the evaluation matrix~$\beta(u)$.
The generalized boundary problem~\eqref{BP:ExSingPlus} for example embeds into the standard boundary problem
\begin{equation} \label{Eq:Embed}
   \bvp{u'' = f,}{u'(0)+u'(1)-2\,u(1)= u'(0)+ u'(1)=0.}
\end{equation}

The Green's operator for this regular problem according to Section \ref{sec:reg} is given by
$x\cum - \cum x - \frac{1}{2} (x+1) + \cum_0^1 	x$. The next proposition tells us how to compute the generalized Green' s operator from it.
\begin{proposition}
  Let $(T, \B, \C)$ be a regular generalized boundary problem and let
  $(T, \tilde{\B})$ be a regular boundary problem with $ \tilde{\B} \leq
  \B$. Then $$(T, \B, \C)^{-1} = (T, \tilde\B)^{-1} \,
  Q,$$ where $Q$ is the projector onto $T(\Bo)$ along $\C$.
\end{proposition}

\begin{proof}
  Since $\tilde{\B} \leq \B$, we have $\Bo \leq \tilde{\B}^{\perp}$. Hence the maps $T|_{\Bo}^{-1}$ and
  $\tilde{G} = (T, \tilde\B)^{-1}$ coincide on $\Bo$. Since $T|_{\Bo} \colon \Bo \to T(\Bo)$ is a bijection, we can compute   the restriction $T|_{\Bo}^{-1}$ by first applying a projector onto $T(\Bo)$ and then $\tilde{G}$. Hence
  $T|_{\Bo}^{-1} = \tilde{G} Q$, where $Q$ is again the projection onto $T(\Bo)$ along $\C$. Hence the generalized
  Green's operator is given by $G = T|_{\Bo}^{-1} Q = \tilde{G}Q^2 =
  \tilde{G}Q$. \qed
\end{proof}

Applying the previous proposition to Example \eqref{BP:ExSingPlus}
leads to the generalized Green's operator $x \cum - \cum x -
\frac{1}{2}(x^2+1) \cum_0^1 + \cum_0^1 x$. For a more involved example
illustrating the \maple\ functions in our package, we refer to the
Appendix.

\section{Outlook}

We are currently investigating in how far the composition of boundary
problems~\eqref{Eq:BComp} can be extended to generalized boundary
problems such that an analog of the ``reverse order law''~\eqref{Eq:GComp} holds. We can see in
the example below~\eqref{embedding} that for such a generalization, we
also have to modify the second component with the boundary
conditions. The question under which conditions a reverse order law
holds for different classes of generalized inverses---not necessarily
related to integro-differential operators---is extensively studied in
the literature, see for example \cite{Djordjevic2007} and the
references therein.

The search for generalized composition laws is intimately connected
with the question of ``embedding'' a singular boundary problem into a
regular problem of higher order.
For example
in~\cite{Rosenkranz2005}, the Green's operator~$G$ of the generalized
boundary problem $(\der^2, [\E_0\der, \E_1\der, \smash{\cum_0^1}],
[1])$ can be factored as $G = \tilde{G} \circ \der$ where $\tilde{G}$
is the standard Green's operator of the boundary problem $(\der^3,
[\E_0\der, \E_1\der, \smash{\cum_0^1}])$. Hence $\tilde{G} = G \circ
\cum_0^x$ and, assuming~\eqref{Eq:GComp} for the composition, also
\begin{equation} \label{embedding}
  (\der^3, [\E_0\der, \E_1\der, \smash{\cum_0^1}], [0]) = (\der, [\E_0],
  [0]) \circ (\der^2, [\E_0\der, \E_1\der,
  \smash{\cum_0^1}], [1]),
\end{equation}
since $\smash{\cum_0^x}$ is the Green's operator of the
boundary problem $(\der, [\E_0])$. The singular second-order problem
is thus embedded into a regular third-order one.

 Multi-point boundary problems can also be treated by our method,
yielding a suitable Green's operator just as in the classical
two-point setting. Generalizing the extraction procedure for Green's
functions is future work, see \cite{Agarwal1986} for
an analytic description of Green's functions for multi-point
boundary problems.

Going from LODEs to LPDEs, more drastic changes are necessary since
geometry enters the picture. For example, the Green's operator of the
inhomogeneous wave equation $u_{xx} - u_{tt} = f$ with homogeneous
Dirichlet data on the~$x$-axis integrates $f$ over a certain triangle
whose tip is at~$(x,t)$. In terms of the operator algebra, this means
one must incorporate the chain and substitution rule along with
explicit operators encoding change of variables. A first approach
along these lines, for the very simple case of linear coordinate
changes, was presented
in~\cite{RosenkranzRegensburgerTecBuchberger2009} and is currently
being refined. Studying singular boundary problems for LPDEs from a
symbolic point of view is also very interesting; see for
example~\cite{KrupchykTuomela2006} for a Gr\"obner bases approach to
compute the (hierarchy of) compatibility conditions for elliptic
boundary problems. It would be tempting to combine the tools of
involutive systems used there with the setting of operator rings used
here.


\begin{thebibliography}{10}

\bibitem{KorporalRegensburgerRosenkranz2010}
Korporal, A., Regensburger, G., Rosenkranz, M.:
\newblock A {M}aple package for integro-differential operators and boundary
  problems.
\newblock ACM Commun. Comput. Algebra \textbf{44}(3) (2010)  120--122 Also
  presented as a poster at ISSAC '10.

\bibitem{RosenkranzRegensburgerTecBuchberger2009}
Rosenkranz, M., Regensburger, G., Tec, L., Buchberger, B.:
\newblock A symbolic framework for operations on linear boundary problems.
\newblock In Gerdt, V.P., Mayr, E.W., Vorozhtsov, E.H., eds.: Computer Algebra
  in Scientific Computing. Proceedings of the 11th International Workshop (CASC
  2009). Volume 5743 of LNCS., Berlin, Springer (2009)  269--283

\bibitem{CoddingtonLevinson1955}
Coddington, E.A., Levinson, N.:
\newblock Theory of ordinary differential equations.
\newblock McGraw-Hill Book Company, Inc., New York-Toronto-London (1955)

\bibitem{Stakgold1979}
Stakgold, I.:
\newblock Green's functions and boundary value problems.
\newblock John Wiley \& Sons, New York (1979)

\bibitem{Duffy2001}
Duffy, D.G.:
\newblock Green's functions with applications.
\newblock Studies in Advanced Mathematics. Chapman \& Hall/CRC, Boca Raton, FL
  (2001)

\bibitem{Agarwal2008}
Agarwal, R.P., O'Regan, D.:
\newblock An introduction to ordinary differential equations.
\newblock Universitext. Springer, New York (2008)

\bibitem{Loud1970}
Loud, W.S.:
\newblock Some examples of generalized {G}reen's functions and generalized
  {G}reen's matrices.
\newblock SIAM Rev. \textbf{12} (1970)  194--210

\bibitem{Ben-Israel2003}
Ben-Israel, A., Greville, T.N.E.:
\newblock Generalized inverses. Second edn.
\newblock Springer-Verlag, New York (2003)

\bibitem{Boichuk2004}
Boichuk, A.A., Samoilenko, A.M.:
\newblock Generalized inverse operators and {F}redholm boundary-value problems.
\newblock VSP, Utrecht (2004)

\bibitem{Nashed1976}
Nashed, M.Z., Rall, L.B.:
\newblock Annotated bibliography on generalized inverses and applications.
\newblock In: Generalized inverses and applications.
\newblock Academic Press, New York (1976)  771--1041

\bibitem{RosenkranzRegensburger2008a}
Rosenkranz, M., Regensburger, G.:
\newblock Solving and factoring boundary problems for linear ordinary
  differential equations in differential algebras.
\newblock J. Symbolic Comput. \textbf{43}(8) (2008)  515--544

\bibitem{RosenkranzRegensburgerTecBuchberger2011}
Rosenkranz, M., Regensburger, G., Tec, L., Buchberger, B.:
\newblock Symbolic analysis for boundary problems: {F}rom rewriting to
  parametrized {G}r{\"o}bner bases.
\newblock In Langer, U., Paule, P., eds.: {Numerical and Symbolic Scientific
  Computing: Progress and Prospects}.
\newblock SpringerWienNew York, Vienna (2011) To appear.

\bibitem{GuoKeigher2008}
Guo, L., Keigher, W.:
\newblock On differential {Rota-Baxter} algebras.
\newblock J. Pure Appl. Algebra \textbf{212}(3) (2008)  522--540

\bibitem{BrownKrall1974}
Brown, R.C., Krall, A.M.:
\newblock Ordinary differential operators under {S}tieltjes boundary
  conditions.
\newblock Trans. Amer. Math. Soc. \textbf{198} (1974)  73--92

\bibitem{RegensburgerRosenkranzMiddeke2009}
Regensburger, G., Rosenkranz, M., Middeke, J.:
\newblock A skew polynomial approach to integro-differential operators.
\newblock In May, J.P., ed.: Proceedings of ISSAC '09, New York, NY, USA, ACM
  (2009)  287--294

\bibitem{RegensburgerRosenkranz2009}
Regensburger, G., Rosenkranz, M.:
\newblock An algebraic foundation for factoring linear boundary problems.
\newblock Ann. Mat. Pura Appl. (4) \textbf{188}(1) (2009)  123--151

\bibitem{Kamke1967}
Kamke, E.:
\newblock Differentialgleichungen. {L}\"osungsmethoden und {L}\"osungen. {T}eil
  {I}: {G}ew\"ohnliche {D}ifferentialgleichungen.
\newblock Akademische Verlagsgesellschaft, Leipzig (1967)

\bibitem{Rosenkranz2005}
Rosenkranz, M.:
\newblock A new symbolic method for solving linear two-point boundary value
  problems on the level of operators.
\newblock {J}. {S}ymbolic {C}omput. \textbf{39}(2) (2005)  171--199

\bibitem{Djordjevic2007}
Djordjevi{\'c}, D.S.:
\newblock Further results on the reverse order law for generalized inverses.
\newblock SIAM J. Matrix Anal. Appl. \textbf{29}(4) (2007)  1242--1246

\bibitem{Agarwal1986}
Agarwal, R.P.:
\newblock Boundary value problems for higher order differential equations.
\newblock World Scientific Publishing Co. Inc., Teaneck, NJ (1986)

\bibitem{KrupchykTuomela2006}
Krupchyk, K., Tuomela, J.:
\newblock The {S}hapiro-{L}opatinskij condition for elliptic boundary value
  problems.
\newblock LMS Journal of Computation and Mathematics \textbf{9} (2006)
  287--329

\end{thebibliography}

\newpage

\appendix

\section{Example}
Now we will give a detailed example for computations with generalized boundary problems. We introduced a new datatype $\gbp(T, \bc(b1, \ldots, bm), \es(f1, \ldots, fk))$, where $T$ and $(b1, \ldots, bm)$ again are a differential operator and boundary conditions and $(f1, \ldots, fm)$ is a basis of the exceptional space. We added the new procedures \textit{CompatibilityConditions, IsComplement} and \textit{Projector}, which will be explained later and extended the procedures \textit{GreensOperator} and \textit{IsRegular}. The first one now also computes the Green's Operator for a generalized boundary problem and the second one tests the condition $\Ker T \cap \Bo = \{0\}$ also for generalized boundary problems.

We consider the more complicated example
\begin{equation}
  \bvp{u'''' + u'' = f}{u'(0)= u''(0)=u''(\pi) = u'''(0) = u'''(\pi)=0.}
\end{equation}
We enter the boundary problem stated above and compute a fundamental system for the differential operator $T= \D{4} + \D{2}$.
\begin{center}
\begin{minipage}{0.999\textwidth}
\begin{framed}~\\[-1cm]
\begin{alltt}\small
> T := DIFFOP(0, 0, 1, 0, 1):
> b[1] := BOUNDOP(EVOP(0, EVDIFFOP(0, 1), EVINTOP())):
> b[2] := BOUNDOP(EVOP(0, EVDIFFOP(0, 0, 1), EVINTOP())):
> b[3] := BOUNDOP(EVOP(0, EVDIFFOP(0, 0, 0, 1), EVINTOP())):
> b[4] := BOUNDOP(EVOP(Pi, EVDIFFOP(0, 0, 1), EVINTOP())):
> b[5] := BOUNDOP(EVOP(Pi, EVDIFFOP(0, 0, 0, 1), EVINTOP())):
> Bp := BP(T, BC(b[1],b[2],b[3],b[4],b[5])):
> fs := FundamentalSystem(T);
\end{alltt}
\begin{equation*}
  [x,\; \sin(x), \;\cos(x), \;1]
\end{equation*}
~\\[-0.9cm]
\end{framed}\normalsize
\end{minipage}
\end{center}
Now we add another boundary condition $\operatorname{b[6]}$ in order to achieve uniqueness of solutions. This can be checked by considering the column rank of the evaluation matrix. We further verify that the compatibility conditions of both problems are the same.
\begin{center}
\begin{minipage}{0.999\textwidth}
\begin{framed}~\\[-1cm]\small
\begin{alltt}
> b[6] := BOUNDOP(EVOP(Pi, ZEROEDOP, EVINTOP(EVINTTERM(1,1)))):
> BpA := BP(T, BC(b[1],b[2],b[3],b[4],b[5],b[6])):
> IsRegular(BpA);
                             \[       true \] \negskip
> CompatibilityConditions(Bp);
\end{alltt}
\begin{equation*}
   \BC((\Ev{\operatorname{Pi}}\,.\,\A\,.\,(\sin(x)), \;(\Ev{\operatorname{Pi}})\,.\,\A\,.\,(\cos(x)))
\end{equation*}
\begin{alltt}
> CompatibilityConditions(BpA);
\end{alltt}
\begin{equation*}
   \BC((\Ev{\operatorname{Pi}}\,.\,\A\,.\,(\sin(x)), \;(\Ev{\operatorname{Pi}})\,.\,\A\,.\,(\cos(x)))
\end{equation*}
~\\[-0.9cm]
\end{framed} \normalsize
\end{minipage}
\end{center}
Now we enter a generalized boundary problem and check that our choice
$[1,\;x]$ as exceptional space is a complement of
$T({\Bo})$. Then we compute the projector $Q$ onto $T(\Bo)$ and the
Green's operator for the generalized boundary problem $(T,
[\operatorname{b[1]}, \operatorname{b[2]}, \operatorname{b[3]},
\operatorname{b[4]}, \operatorname{b[5]}, \operatorname{b[6]}],
[1,x])$,
\begin{center}
\begin{minipage}{0.999\textwidth}
\begin{framed}~\\[-1cm]\small
\begin{alltt}
> gBp := GBP(T, BC(b[1],b[2],b[3],b[4],b[5],b[6]), ES(1,x)):
> IsComplement(gBp);
                               \[       true \] \negskip
> Q := Projector(gBp):
\end{alltt}
\begin{equation*}
  Q := 1 - \frac{1}2 ((\Ev{\operatorname{Pi}})\,.\, A\, .\, (\sin(x))) + \Biggl(\Bigl(-\frac{\operatorname{Pi}}4 + \frac{x}2 \Bigr)\, .\, (\Ev{\operatorname{Pi}})\, .\, A\, .\, (\cos(x) \Biggr)
\end{equation*}
\begin{alltt}
> G := GreensOperator(gBp):
\end{alltt}
\end{framed}\normalsize
\end{minipage}
\end{center}
Finally we verify that the Green's operator $G$ fulfills the equation
$TG=Q$ and the six boundary conditions.
\begin{center}
\begin{minipage}{0.999\textwidth}
\begin{framed}~\\[-1cm] \small
\begin{alltt}
> simplify(SubtractOperator(MultiplyOperator(T, G), Q))
                                   \[    0  \] \negskip
> seq(simplify(MultiplyOperator(b[i], G)), i=1..6);
                                   \[ 0, 0, 0, 0, 0, 0  \] \negskip
\end{alltt} \normalsize
~\\[-1cm]
\end{framed}
\end{minipage}
\end{center}

\end{document}